\documentclass[12pt]{article}

\usepackage[margin=1in]{geometry}
\usepackage{setspace}

\usepackage{amsmath,amssymb,amsthm}
\usepackage{bm}
\usepackage{graphicx}
\usepackage{color}
\usepackage{hyperref}

\newtheorem{theorem}{Theorem}
\newtheorem{lemma}{Lemma}

\newcommand{\Z}{\mathbb{Z}}
\newcommand{\N}{\mathbb{N}}
\newcommand{\J}{\mathbb{J}}
\newcommand{\F}{\mathbb{F}}
\newcommand{\R}{\mathbb{R}}
\renewcommand{\S}{\mathbb{S}}
\newcommand{\T}{\mathbb{T}}

\newcommand{\Cov}{\mathrm{Cov}}

\usepackage[round]{natbib}

\begin{document}


\begin{center}{\LARGE Spectral Density Estimation for Random Fields  \\

\vspace{6pt}

via Periodic Embeddings }

\vspace{12pt}

{Joseph Guinness}
\vspace{12pt}

\textit{Cornell University, Department of Statistics and Data Science}
\vspace{24pt}

\textbf{Abstract}

\end{center}

\noindent We introduce methods for estimating the spectral density of a random field on a $d$-dimensional lattice from incomplete gridded data. Data are iteratively imputed onto an expanded lattice according to a model with a periodic covariance function. The imputations are convenient computationally, in that circulant embedding and preconditioned conjugate gradient methods can produce imputations in $O(n\log n)$ time and $O(n)$ memory. However, these so-called periodic imputations are motivated mainly by their ability to produce accurace spectral density estimates. In addition, we introduce a parametric filtering method that is designed to reduce periodogram smoothing bias. The paper contains theoretical results studying properties of the imputed data periodogram and numerical and simulation studies comparing the performance of the proposed methods to existing approaches in a number of scenarios. We present an application to a gridded satellite surface temperature dataset with missing values.

\section{Introduction}

Random fields defined on the integer lattice have wide applications for modeling gridded spatial and spatial-temporal datasets. They also form the basis for some models for non-gridded data \citep{nychka2015multiresolution}. The large sizes of modern spatial and spatial-temporal datasets impart an enormous computational burden when using traditional methods for estimating random field models. Modeling data on a grid provides a potential solution to the computational issue, since there exist some methods based on the discrete Fourier transform, which can be computed efficienty with fast Fourier transform algorithms. However, there are some pitfalls associated with discrete Fourier transform-based methods related to edge effects and the handling of missing data. This paper provides an accurate and computationally efficient estimation framework for addressing those issues.

Let $Y(x) \in \R$, $x\in \Z^d$ be a mean-zero stationary process on the $d$-dimensional integer lattice, that is, $E\{Y(x)\} = 0$ and $\Cov\{Y(x),Y(x+h)\} = K(h)$ for every $x$ and $h$ in $\Z^d$. Herglotz's theorem states that the covariance function has a Fourier transform representation,
\begin{align}\label{covariance}
K(h) = \int_{[0,1]^d} \exp(2\pi i\omega \cdot h) dF(\omega),
\end{align}
where $i = \sqrt{-1}$ and $\cdot$ is the dot product. The function $F$ is a spectral measure, and we assume throughout that it has a continuous derivative $f$ called a spectral density. We focus on estimation of $f$, which encodes the covariance function, and thus is crucial for prediction of missing values and for regressions when $Y$ is used as a model for residuals. We restrict our attention to stationary models and note that stationary models often form the basis for more flexible nonstationary models that are needed to accurately model many physical processes \citep{fuentes2002spectral}.

Suppose that we observe vector $U = \{Y(x_1),\ldots,Y(x_{n})\}$ at a distinct set of $n$ locations $\mathbb{S}_1 = \{ x_1,\ldots,x_{n}\}$. If $f$ or $K$ have known parametric forms, and we assume that $Y(x)$ is a Gaussian process, then we can use likelihood-based methods for estimating the parameters, which generally requires $O(n^2)$ memory and $O(n^3)$ floating point operations. If the locations form a complete rectangular subset of the integer lattice, we can use Whittle's likelihood approximation \citep{whittle1954stationary}, which leverages fast Fourier transform algorithms in order to approximate the likelihood in $O(n \log n)$ FLOPs and $O(n)$ memory. \cite{guyon1982parameter} showed that, due to edge effects, the Whittle likelihood parameter estimates are not root-$n$ consistent when the dimension $d$ of the field is greater than 1. \cite{dahlhaus1987edge} suggested the use of data tapers to reduce edge effects and proved that the tapered version of the likelihood approximation is asymptotically efficient when $d \leq 3$. \cite{stroud2017bayesian} and \cite{guinness2016circulant} suggested the use of periodic embeddings and demonstrated their accuracy in numerical studies. \cite{sykulski2016biased} introduced a de-biased Whittle likelihood.

If one is not willing to assume that $f$ or $K$ have known parametric forms, and if the data are observed on a complete rectangular grid, nonparametric methods can be used to estimate $f$. The standard approach uses the discrete Fourier transform,
\begin{align}\label{periodogram}
J(\omega) = \frac{1}{\sqrt{n}} \sum_{j=1}^n Y(x_j)\exp(-2\pi i\omega \cdot x_j),
\end{align}
and estimates the spectrum with a smoothed version of the periodogram $|J(\omega)|^2$,
\begin{align*}
\widehat{f}(\omega) = \sum_{\nu} |J(\nu)|^2 \alpha(\omega-\nu),
\end{align*}
where $\alpha$ is a smoothing kernel. Selection of the kernel bandwidth has been studied by several authors \citep{lee1997simple,ombao2001simple,lee2001stabilized}. Alternatively, one can smooth using penalized likelihoods \citep{wahba1980automatic,chow1985sieve,pawitan1994nonparametric} or smooth priors in a Bayesian setting \citep{zheng2009nonparametric}. \cite{politis1995bias} provided a method for reducing bias in the smoothed periodogram. \cite{heyde1993smoothed} studied asymptotic properties of the periodogram in an increasing domain setting, while \cite{stein1995fixed} studied them in an increasing resolution setting, noting the importance of data filtering. \cite{lim2008properties} considered the multivariate case.

The nonparametric methods discussed above apply when a complete dataset is available on a rectangular grid. However, even when available on a grid, spatial datasets often have many missing values; for example, it is common to encounter gridded satellite datasets with some values obscured by clouds. Missing values complicate two aspects of periodogram-based estimators. The first is that a surrogate for the missing values must be substituted; Fuentes (2007, Section 3) suggested replacing them with zeros and scaling the periodogram by the number of observed grid cells. Also of relevance is the extensive theoretical literature on spectral domain analysis for irregularly sampled spatial data  \citep{matsuda2009fourier,bandyopadhyay2009asymptotic,bandyopadhyay2015frequency,deb2017asymptotic,subbarao2018}, which can be applied to incomplete gridded data as well. All of these approaches use a discrete Fourier transform of the sampled data,
which for gridded datasets is equivalent to the zero-infill approach in  Fuentes (2007, Section 3). Numerical comparisons between a zero-infill approach and our new approach are given in Section \ref{numericalsection}. A second problem for spatial data is that scattered missing values seriously disrupt the use of differencing filters. For example two-dimensional differencing at observed location $(j,k)$ can be applied only if observations at $(j+1,k)$, $(j,k+1)$, and $(j+1,k+1)$ are observed as well.

To address these issues, this paper introduces computationally efficient methodology for estimating the spectrum based on imputing missing values with conditional simulations and iteratively updating the spectrum estimate, in a similar vein as \cite{lee2009nonparametric} proposed for time series data. The novelty of our approach is that the missing values are imputed onto an expanded lattice under a covariance function that is periodic on the expanded lattice. These {periodic imputations} or {periodic conditional simulations} are convenient computationally, since circulant embedding and preconditioned conjugate gradient methods can be employed for efficient imputations, but their main appeal is their ability to produce accurate estimates via the amelioration of edge effects. We provide thorough numerical studies and theoretical results describing when the imputed data spectrum is expected to give an estimate with a smaller bias than the spectrum used for imputation, which suggests that existing spectral density estimates can be improved through periodic embedding.

The theoretical results provide a sound basis for the nonparametric estimation methods and give some insight into why the parametric methods in \cite{guinness2016circulant} perform so well in simulations. Additionally, this paper introduces a parametric filtering method based on fitting simple parametric models within the iterative method. The fitted parametric models can be used to filter the data, which is effective for reducing bias due to periodogram smoothing. Taken together, this work results in accurate and computationaly efficient methods for estimating spectral densities when the gridded data have arbitrary missingness patterns. We provide thorough numerical and simulation studies for the methods and demonstrate that even a small amount of lattice expansion provides substantial bias and correlation reduction. We apply the methods to a gridded but incomplete land surface temperature dataset.

\section{Methodology}\label{methodssection}

\subsection{Notation and Background}\label{background}

Let $y = (y_1,\ldots,y_d)$ with $y_i \in \N$, and define the hyperrectangle $\J_y \subset \Z^d$, where
\begin{align*}
\J_y = \{ (a_1,\ldots,a_d) \mid a_j \in \{1,\ldots,y_j\} \,\, \mbox{for all } j\}.
\end{align*}
If $d=2$, this is simply a rectangular lattice of size $(y_1,y_2)$. We assume that the observation locations $\S_1$ form a subset of $\J_y$, and so we call $\J_y$ the observation lattice.
Define $V$ to be the vector containing the process at the remaining locations $\J_y \setminus \S_1$. Throughout, we asssume that $Y(x)$ is missing at random, meaning that the missingness is potentially related to $x$ but not related to the value of $Y(x)$  \citep{little2014statistical}. This section describes several existing and new iterative methods for estimating a spectral density $f$. All of these methods proceed by updating the spectrum estimate at the $k$th iteration, $f_k$, to the next estimate, $f_{k+1}$. Although the specific updating formulas vary, we use the notation $f_{k+1}$ for all of them to keep the number of symbols manageable.


For time series data, \cite{lee2009nonparametric} proposed an iterative method for obtaining nonparametric estimates of the spectrum. Let ${E}_k$ denote expectation in the mean zero multivariate normal distribution for $(U,V)$ under $f_k$ with covariance given by \eqref{covariance}. Their method can be extended from one dimension to general dimensions as follows with the updating formula
\begin{align}\label{update1}
f_{k+1}(\omega) = \sum_{\nu \in \F_y} {E}_k \Big\{ |J(\nu)|^2 \, \big| \, U \Big\}\alpha(\omega - \nu), \end{align}
where $\F_y$ is the set of Fourier frequencies associated with a grid of size $y$. The procedure is then iterated over $k$ until convergence. Here, we use a smoothing kernel, but \cite{lee2009nonparametric} noted that any smoothing method can be applied. The conditional expectation of the periodogram under $f_k$ is computationally expensive, so
\cite{lee2009nonparametric} proposed replacing the expected value with an average over $L$ independent realizations of $V$ given $U$, as in
\begin{align}\label{iterator1}
f_{k+1}(\omega) = \sum_{\nu \in \F_y} \frac{1}{L} \sum_{\ell = 1}^L  |J^{(\ell)}(\nu)|^2  \alpha(\omega - \nu),
\end{align}
where $J^{(\ell)}$ is the discrete Fourier transform derived from $(U,V^{(\ell)})$, where $V^{(1)},\ldots,V^{(L)}$ are independent Gaussian conditional simulations of $V$ given $U$ under $f_k$.
Replacing the conditional expectation with a sample average is analogous to the approach taken in the iterative method in \cite{tanner1987calculation} for Bayesian estimation of parametric statistical models. In this case, using a sample average creates a convergence issue, in that the Monte Carlo error causes the spectra in \eqref{iterator1} to fluctuate indefinitely. In Subsection \ref{periodicimputationsection}, we propose an alternative averaging scheme, as well as imputation under a periodic model.

\subsection{Periodic Imputation}\label{periodicimputationsection}

When $d>1$, edge effects become a prominent issue \citep{guyon1982parameter}; in particular, the Whittle likelihood can be interpreted as the exact likelihood for a model in which the field is periodic on the observation lattice \citep{guinness2016circulant}. Data tapers have been proposed to alleviate the issue, but tapering can lead to loss of information from data near the boundaries or near missing values. In this paper, we propose extending the hyperrectangle in each dimension and performing the imputations under a periodic approximation to the covariance function. Surprisingly, using the periodic approximation to the covariance function for the imputations, rather than the true covariance function, leads to improved spectral density estimates. This is demonstrated numerically in Section \ref{numericalsection}. Periodic models also facilitate straightforward implementation of circulant embedding techniques to simulate from the conditional distributions efficiently.

Let $\tau \geq 1$, and define $z_j = \lceil \tau y_j \rceil$ so that $z_j \geq y_j$ for $j = 1,\ldots,d$. Define $m=z_1 \cdots z_d$ to be the total number of locations in $\J_z$, which we refer to as the embedding lattice. Let $W$ denote the vector of missing values on $\J_z \setminus \S_1$ and $\widetilde{E}_k$ denote expectation in the mean-zero multivariate normal distribution for $(U,W)$ with covariance function $R_k(\cdot)$, defined as
\begin{align}\label{periodiccovariance}
R_k(h) = \frac{1}{m} \sum_{\omega \in \F_z} f_k(\omega) \exp(2\pi i\omega \cdot h), \quad h = (h_1,\ldots,h_d),
\end{align}
where $\F_z$ are the Fourier frequencies associated with $\J_z$. Note that for every $\omega \in \F_z$, the function $\exp(2\pi i \omega \cdot h)$ is periodic in $h_j$ with period $z_j$. This ensures that $R_k(\cdot)$ is periodic on $\J_z$ in each dimension and is not the integral Fourier transform of $f_k$ that appears in \eqref{covariance}. We refer to a draw of $W$ under $R_k(\cdot)$ as a periodic conditional simulation or a periodic imputation. Figure \ref{periodicsim} contains an example with $\tau = 1.15$.

\begin{figure}
\centering
\includegraphics[width=0.8\textwidth]{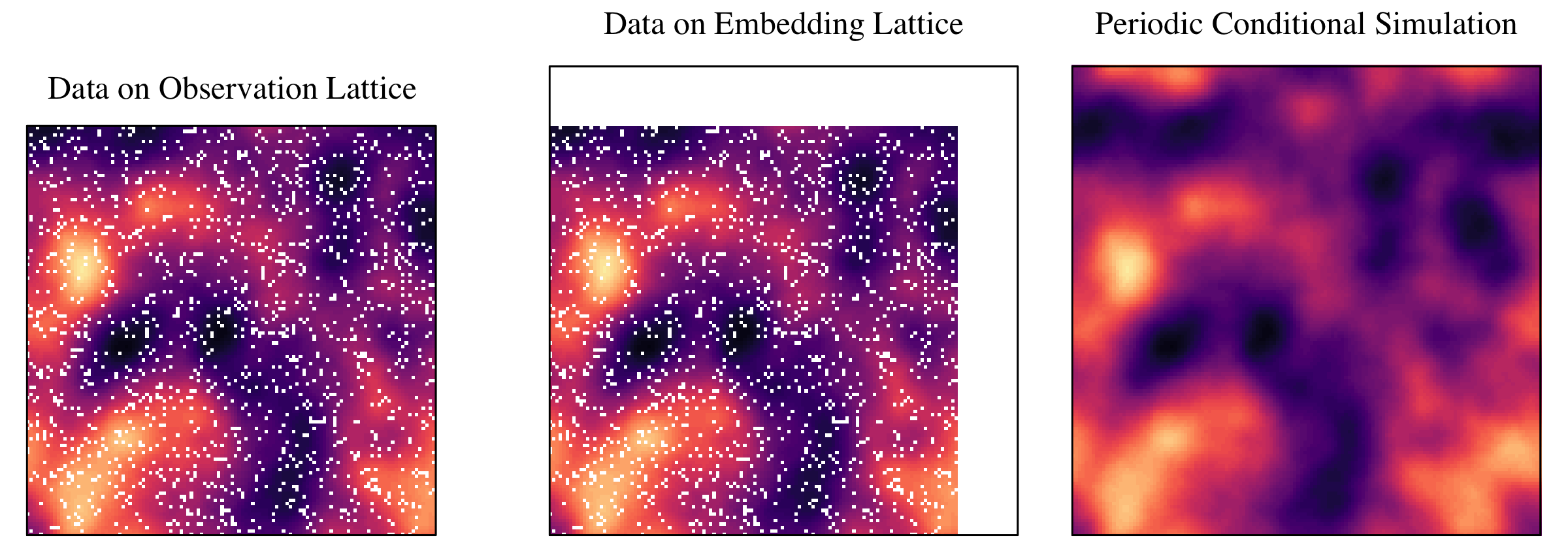}
\caption{Data on the observation lattice $\J_y$, the embedding lattice $\J_z$, and a periodic conditional simulation. \label{periodicsim} }
\end{figure}

Using conditional expectations, the update in the periodic model is
\begin{align}\label{update3}
f_{k+1}(\omega) = \sum_{\nu \in \F_z} \widetilde{E}_k \Big\{ |J(\nu)|^2 \, \mid \, U \Big\} \alpha(\omega - \nu).
\end{align}
Note that the conditional expectation in the \cite{lee2009nonparametric} estimator in Equation \eqref{update1} is done on the observation lattice and using the correct model, whereas in \eqref{update3} we use the conditional expectation under a model that is periodic on the embedding lattice. As before, the conditional expectation can be replaced by the average over one or several conditional simulations. To address the convergence issue mentioned in Subsection \ref{background}, we propose an alternative updating formula consisting of a burn-in period of $B$ iterations and convergence monitoring based on the asymptotic standard deviation of the complete data smoothed periodogram,
\begin{align*}
S_k(\omega) = \Big\{ \sum_{\nu \in \F_z} f_{k}(\nu)^2 \alpha(\omega - \nu)^2 \Big\}^{1/2}.
\end{align*}

Our full proposed estimation algorithm is as follows. Initialize $f_1(\omega)$ constant (flat spectrum), and given spectrum $f_k$, update as follows:
\begin{enumerate}
\item For $\ell = 1,\ldots,L$, conditionally simulate $W^{(\ell)}$ given $U$ under $f_k$,
\item For $\ell = 1,\ldots,L$, compute $J^{(\ell)}(\omega)$ from $(U,W^{(\ell)})$,
\item Update spectrum as
\begin{align}
\hspace{-0cm} f_{k+1}(\omega) = \left\lbrace \begin{array}{ll}
\displaystyle \frac{1}{L} \sum_{\ell=1}^L \sum_{\nu \in \F_z} |J^{(\ell)}(\nu)|^2 \alpha(\omega - \nu) & \quad k \leq B \\
\frac{k-B}{k-B+1} \displaystyle {f}_{k}(\omega) +  \textstyle \frac{1}{k-B+1} \displaystyle\frac{1}{L}  \sum_{\ell=1}^L \sum_{\nu \in \F_z} |J^{(\ell)}(\nu)|^2 \alpha(\omega - \nu) & \quad k > B .
\end{array} \right.
\end{align}
\end{enumerate}
The algorithm is stopped when
\begin{align*}
\max_{\omega \in \F_z} \frac{|{f}_{k+1}(\omega) - {f}_{k}(\omega)|}{S_{k}(\omega)} < \varepsilon.
\end{align*}

To summarize, during the $B$ burn-in iterations, we use the sample average version of Equation \eqref{update3}. After burn-in, the updating formula uses a weighted average of the previous spectrum and the current smoothed periodogram. Using a burn-in period avoids averaging over spectra from the first few iterations. Convergence is relative to the asymptotic standard deviation of the complete data smoothed spectrum and a tolerance criterion $\varepsilon$, that we take to be $0.05$ or $0.01$ in practice. We typically take $L=1$ in practice. Appendix \ref{circulantappendix} contains details on how circulant embedding and preconditioned conjugate gradient can be employed to efficiently compute the periodic conditional simulations.

\subsection{Variant with parametric filter}

Even if $|J(\omega)|^2$ is unbiased for $f(\omega)$, the smoothing step can introduce some bias in the spectral density estimate. For spectral densities with large dynamic range, data filters have been proposed to pre-whiten the data prior to smoothing \citep{stein1995fixed}. Missing data pose a challenge for data filters, but filters can be easily applied to the imputed data at each iteration. In this subsection, we propose a parametric filtering method that we show in simulations is successful in reducing smoothing bias.

Let $f_\theta$ be a parametric spectral density. The imputed data Whittle likelihood approximation is
\begin{align*}
\ell(\theta) = -\frac{m}{2}\log 2\pi - \frac{1}{2}\sum_{\omega \in \F_z} \left[ \log f_\theta(\omega) + \frac{\widetilde{E}_k\{ | J(\omega)|^2 \, | \, U\}}{f_\theta(\omega)} \right].
\end{align*}
Let $\theta_k$ be the maximizer of $\ell(\theta)$. Then update as
\begin{align}
{f}_{k+1}(\omega) = f_{\theta_k}(\omega) \sum_{\nu \in \F_z} \frac{\widetilde{E}_k\{|J(\nu)|^2 \, | \, U\}}{f_{\theta_k}(\nu)} \alpha(\omega - \nu).
\end{align}
As before, in practice we replace $\widetilde{E}_k\{|J(\nu)|^2 \, | \, U\}$ with a sample average that can be computed efficiently. The completely nonparametric variant is a special case with $f_\theta(\omega)$ constant. Using the parametric step in the smoothing serves to flatten the periodogram, which we show in simulation studies is helpful for reducing smoothing bias. This allows for the use of wider smoothing kernels, which reduces variance as well.

The parametric Mat\'ern covariance is a popular choice for modeling spatial data, and so we recommend using some form of the Mat\'ern for the parametric model. \cite{guinness2016circulant} described a quasi-Mat\'ern covariance, whose spectral density can be evaluated quickly without aliasing calculations. Based on their results, we recommend using the quasi-Mat\'ern in practice. A special case of the quasi-Mat\'ern is explored in Section \ref{numericalsection}.

\section{Theory}\label{theorysection}

This section studies bias in the imputed data periodogram and correlation in the imputed data discrete Fourier transform vector. We use the notation that $f$ is the true spectrum, and $f_1$ is a spectrum to be used for imputation. The theorem should be interpreted as a statement about how the discrete Fourier transform vector behaves given a particular imputation spectrum, not about the iterative procedure itself. Section \ref{numericalsection} contains a numerical exploration of the iterative procedure, and Section \ref{discussionsection} discusses issues related to the theoretical study of the iterative procedure.

Let $R$ (without parentheses) be the covariance matrix for $(U,W)$ under periodic covariance function $R(\cdot)$ in \eqref{periodiccovariance} with spectrum $f$. Partition $R$ as $[A \,\, B; B^T \,\, C]$, so that $A$ and $C$ are the covariance matrices for $U$ and $W$, respectively. Let $K$ denote the covariance matrix for $U$ under the true nonperiodic covariance function $K(\cdot)$ in \eqref{covariance}. Note that $R$ is $m\times m$, while $K$ is $n\times n$. Define $R_1$ to be the covariance matrix for $(U,W)$ under periodic covariance function $R_1(\cdot)$ with spectrum $f_1$, and define $A_1$, $B_1$, and $C_1$ accordingly. Throughout, we assume that both $f$ and $f_1$ are bounded above and below by positive constants. If $W_1$ is a periodic conditional simulation given observations $U$ under $f_1$, then the true covariance matrix for $(U,W_1)$ is
\begin{align}\label{uv1covariance}
S := \Cov\Big\{ \begin{pmatrix} U \\ W_1 \end{pmatrix} \Big\} =
\begin{pmatrix}
K & K A_1^{-1} B_1 \\ B_1^T A_1^{-1} K & \quad C_1 - B_1^T A_1^{-1} B_1 + B_1^T A_1^{-1} K A_1^{-1} B_1
\end{pmatrix}.
\end{align}
The matrix $S$ is a key object of study, and it is of interest to understand its Fourier spectrum. To this end,
define the $m\times 1$ vector $g(\nu)$ to have entries $m^{-1/2}\exp(2\pi i\nu \cdot x)$, where $\nu \in \F_z$ is a Fourier frequency, and $x \in \J_z$, with the entries of $g(\nu)$ ordered as they are in $(U,W)$. Define
\begin{align*}
\widetilde{f}_1(\nu,\omega) := E[\widetilde{E}_1\{J(\nu)J(\omega)^*|U\}] = g(\nu)^\dag S g(\omega),
\end{align*}
where $*$ is complex conjugate, and $\dag$ is conjugate transpose, so that $\widetilde{f}_1(\omega,\omega)$ is the Fourier spectrum of $S$ from which we construct our estimates of the spectrum. Likewise, we define $f(\nu,\omega) = f(\nu)$ if $\omega = \nu$ and $0$ if $\omega \neq \nu$. This notation is useful for succinct theorem statements and reflects the fact that the true bispectrum is zero off the diagonal for stationary models. It is of interest to study $\widetilde{f}_1(\nu,\omega) - f(\nu,\omega)$, which for $\omega = \nu$ corresponds to the bias of the periodogram, and for $\omega \neq \nu$ measures dependence in the periodogram, both of which should ideally be near zero.

The difference $\widetilde{f}_1(\nu,\omega) - f(\nu,\omega)$ will be exactly zero for every $\nu,\omega$ if and only if $S = R$ due to the uniqueness of the Fourier transform. Inspection of \eqref{uv1covariance} suggests that $S$ approaches $R$ if both $K$ and $A_1$ approach $A$. The entries of $K$ come from true covariance function $K(h)$, and the entries of $A$ come from periodic covariance function $R(h)$. To see when $K$ approaches $A$, consider the multidimensional Poisson summation formula,
\begin{align*}
R( h ) = \sum_{j \in \Z^d} K( h + j \circ z ) = K(h) + \sum_{j \neq 0} K( h + j \circ z ),
\end{align*}
where $j \circ z$ is the elementwise product $(j_1 z_1, \ldots, j_d z_d)$. This says that $R(h) - K(h)$--and thus $K - A$--approaches zero whenever $K(h + j\circ z)$ decays quickly enough, which can be ensured by placing smoothness conditions on the spectrum. We now state the main result.



\begin{theorem}\label{theoremC} Let $f$ have $p$ continuous partial derivatives, $y_j = O( n^{1/d} )$, and $z_j = O(\tau y_j)$ for $\tau > 1$. Define $\Delta = \max_{\omega \in [0,1]^d}| f(\omega) - f_1(\omega)|$. Then for every $\nu,\omega \in \F_z$,
\begin{align*}
\widetilde{f}_1(\nu,\omega) - f(\nu,\omega) = O(n^{-p/d+1}) + O\Big\{\Delta \Big(\frac{m-n}{m}\Big)^{1/2}   \Big\},
\end{align*}
meaning that the difference contains two terms with the respective rates.
\end{theorem}

The first term in the rate derives from the decay of the covariances $K(h)$. This term decays quickly with $n$ when the spectrum is smooth, and the dimension of the domain is small. The second concerns the proportion of missing values relative to the number of observed values, which, when small, overwhelms the fact that $A_1 \neq A$. The assumptions about how the observation grid grows with $n$ are standard assumptions that ensure that each dimension grows at the same rate with $n$. When the spectrum is smooth enough, the first decay rate is better than the usual $n^{-1/d}$ \citep{guyon1982parameter} or even $n^{-1}$ rates for the bias of the non-imputed periodograms. The proof is given in Appendix \ref{proofs}, along with intermediate results that assume a correct imputation spectrum.

The implication of the theorem is that when $n$ is large enough and $(m-n)/m$ is small enough, we can initialize the iterative algorithm with any estimate of the spectrum (e.g.\ \cite{fuentes2007approximate,matsuda2009fourier}), and one step in the iterative algorithm will decrease the bias relative to the initialized estimate. The theorem does not make any claims about convergence of the iterative algorithm; these issues are explored numerically in Section \ref{numericalsection}.

\section{Numerical Studies and Simulations}\label{numericalsection}


To provide more insight about the behavior of the proposed estimation methods, we present a numerical study analyzing the bispectrum of the imputed data and simulation results comparing the proposed estimators to other spectral density estimators. The numerical study involves calculations of the bispectrum from covariance matrices and thus involves no simulated data. In the simulation study, we estimate the spectral densities on simulated datasets, which allows us to study sampling variability and the effect of smoothing on the estimated spectral densities. In both the numerical study and the simulation study, we consider data on square grids under three missingness settings shown in Figure \ref{simulation_examples}. The first setting has 30\% scattered missing values. The second setting has a missing block in the center of the grid, with roughly 30\% of the total missing. The third setting has no missing values.

\begin{figure}
\centering
\includegraphics[width=0.8\textwidth]{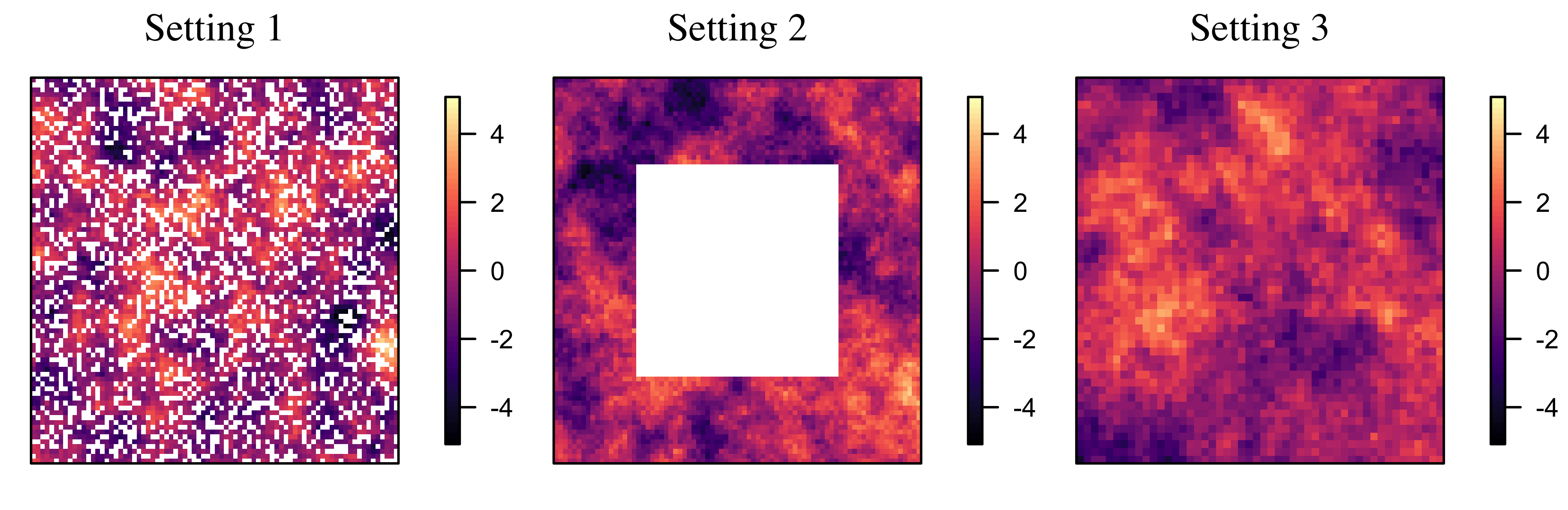}
\caption{Example realizations from the three missingness settings, with missing values in white. }
\label{simulation_examples}
\end{figure}

In the numerical study, we assume that the true covariance function is $K(h) = 2\exp(-\| h\|/8)$, with data on a $(32,32)$ grid.
Let $f_0$ be the bispectrum of $K$, that is,
\begin{align*}
f_0(\nu,\omega) = \frac{m}{n} g(\nu)^\dag \begin{pmatrix} K & 0 \\ 0 & 0 \end{pmatrix} g(\omega).
\end{align*}
Then for $k \geq 0$, let $f_{k+1}(\nu,\omega) = g(\nu)^\dag S_{k} g(\omega)$, where
\begin{align*}
S_k = \begin{pmatrix} K & K A_k^{-1} B_k \\ B_k^T A_k^{-1} K & \quad C_k + B_k^{T} A_{k}^{-1} (K - A_k) A_k^{-1} B_k \end{pmatrix}.
\end{align*}
This numerical study mirrors a setting where we initialize the iterative procedure with the periodogram of the non-imputed data. This is repeated for four values of expansion factor $\tau \in \{32/32, 34/32, 36/32, 38/32\} = \{1,1.0625, 1.125,1.1875\}$. We quantify the error in the bispectrum with an integrated normalized squared bias
\begin{align*}
\frac{1}{m}\sum_{\nu \in \F_z} \sum_{\omega \in \F_z} \frac{ \Big\{{f}_k(\nu,\omega) - f(\nu,\omega) \Big\}^2 }{f(\nu,\nu) f(\omega,\omega)  }.
\end{align*}

The results for the integrated normalized squared bias are shown in Table \ref{numstudytable}. The column for iteration 0 corresponds to bias in the non-imputed data periodogram and has values that are quite large compared to the imputed data periodograms, especially in Setting 1. Rows 1 and 5 correspond to imputation of missing values on the original data domain (row 9 has no missing values). We see that imputing missing values on the original domain offers some improvement. However, imputing on an expanded domain gives biases that orders-of-magnitude smaller in many cases, and the biases decrease substantially in just a few iterations. It is also apparent that even a small amount of expansion lowers the bias; for example, expanding the domain by four pixels $(\tau = 36/32)$ gives biases near zero even though the spatial range parameter is twice as large as the domain expansion.


\begin{table}
\begin{tabular}{ccccccccc}
& & \multicolumn{7}{c}{Iteration} \\
Setting & Expansion $\tau$ & 0 & 1 & 2 & 3 & 4 & 5 & 6 \\
\hline
  1 & 32/32  &  757.6  &  9.584  &  5.946  &  5.457  &  5.457  &  5.516  &  5.562 \\
  1 & 34/32  &  866.4  &  5.077  &  1.181  &  0.406  &  0.230  &  0.185  &  0.173 \\
  1 & 36/32  &  971.6  &  5.663  &  1.466  &  0.432  &  0.152  &  0.069  &  0.043 \\
  1 & 38/32  &  1083.3  &  6.332  &  1.933  &  0.638  &  0.228  &  0.090  &  0.040 \\
\hline
  2 & 32/32  &  27.20  &  8.622  &  8.305  &  8.201  &  8.161  &  8.144  &  8.136 \\
  2 & 34/32  &  24.87  &  0.613  &  0.279  &  0.222  &  0.210  &  0.206  &  0.206 \\
  2 & 36/32  &  27.99  &  0.494  &  0.133  &  0.059  &  0.040  &  0.035  &  0.033 \\
  2 & 38/32  &  31.40  &  0.531  &  0.146  &  0.052  &  0.024  &  0.015  &  0.011 \\
\hline
  3 & 32/32  &  7.990  &  7.990  &  7.990  &  7.990  &  7.990  &  7.990  &  7.990 \\
  3 & 34/32  &  5.489  &  0.231  &  0.201  &  0.200 &  0.200  &  0.200  &  0.200 \\
  3 & 36/32  &  6.297  &  0.083  &  0.035  &  0.031  &  0.031  &  0.031  &  0.031 \\
  3 & 38/32  &  7.180  &  0.079  &  0.016  &  0.010  &  0.009  &  0.009  &  0.009 \\
\end{tabular}
\caption{ \label{numstudytable} Integrated normalized squared bias under exponential covariance model, for three missingness settings, four expansion factors (including no expansion $\tau = 1$), and from 0 to six iterations. }
\end{table}

In the simulation study, we use an $(80,80)$ grid in Settings 1 and 2, and a $(50,50)$ grid in Setting 3. Data are generated from a mean-zero Gaussian process model with Mat\'ern covariance function
\begin{align*}
K(h) = \frac{2}{\Gamma(\nu) 2^{\nu-1}} \Big( \frac{\sqrt{2\nu}\|h\|}{8} \Big)^\nu \mathcal{K}_\nu \Big(\frac{\sqrt{2\nu}\|h\|}{8}\Big),
\end{align*}
with three different choices of smoothness parameter $\nu \in \{1/2,1,3/2\}$, range parameter 8, and variance parameter 2.

We consider several methods for estimating the spectral densities. The first method uses a smoothed periodogram computed from the discrete Fourier transform of the sampled data, scaled by the number of observations $n$. This method is described in the Introduction and is the approach suggested by \cite{fuentes2007approximate,matsuda2009fourier,bandyopadhyay2009asymptotic,bandyopadhyay2015frequency,deb2017asymptotic,subbarao2018}. The second method uses a periodogram computed from tapered data. We define one-dimensional cosine tapers $T_1$ and $T_2$ applied to 5\% of the observations on each of the two edges, and the taper function is the outer product $T((j,k)) = T_1(j)T_2(k)$. In Setting 1, the taper function is set to zero whenever there is a missing value. In Setting 2, which includes a square of missing values in the center, we also taper the interior observations. The periodograms of tapered data are normalized by the sum of the squares of the taper function. Additionally, we consider the \cite{lee2009nonparametric} estimator described in Section \ref{methodssection} (i.e.\ non-periodic imputation) and variants of their method that use lattice expansion and/or parametric filters. Using a non-periodic embedding method allows us to separate the effect of using a larger lattice from the effect of imputing periodically.


For the imputation-based methods proposed in this paper, we consider lattice expansion factors $\tau \in \{1.0,1.1,1.2\}$. We also consider two settings for the use of a parametric filter, the first being no filter, and the second with filter of the form
\begin{align}\label{spec_AR1}
f_\theta(\omega) = \frac{1}{1 - \frac{\theta}{2}\{ \cos(2\pi \omega_1) + \cos(2\pi \omega_2)\}},
\end{align}
where $0 \leq \theta < 1$. This choice for the parametric model is a member of the quasi-Mat\'ern family \citep{guinness2016circulant} and is deliberately misspecified for the two cases $\nu = 1/2$ and $\nu = 3/2$. \cite{lindgren2011explicit} showed that this model can approximate the Mat\'ern covariance with smoothness parameter equal to 1.

All of the imputation-based estimation methods use $L=1$ conditional simulations, $B = 100$ burn-in iterations, and use convergence criterion $\varepsilon = 0.01$. The estimate from the $j$th dataset is denoted by $\widehat{f}^{j}$. All methods use a Gaussian smoothing kernel proportinal to $\exp( - \| \omega - \nu \|^2/\delta^2 )$, where the distance $\| \omega - \nu \|$ is defined periodically on the domain $[0,1]^2$. We consider two metrics for evaluating the estimation methods. The first is a relative bias
\begin{align*}
\mbox{BIAS}(\omega) = \frac{1}{100} \sum_{j=1}^{100} \frac{ \widehat{f}^{j}(\omega) - f(\omega) }{f(\omega)},
\end{align*}
where $100$ is the total number of simulated datasets, and $f$ is the true spectrum. The second metric is a mean relative squared error
\begin{align*}
\mbox{MSE}(\omega) = \frac{1}{100} \sum_{j=1}^{100} \left\{ \frac{ \widehat{f}^{j}(\omega) - f(\omega) }{f(\omega)} \right\}^2.
\end{align*}

To evaluate relative bias on an equal footing, we compare all methods using a small value of $\delta = 0.02$. Figure \ref{biasplots} contains plots of the relative bias for the non-tapered and tapered methods, and non-filtered and filtered periodic and nonperiodic embedding methods with $\tau = 1.2$. Results for $\nu=1/2$ are shown (results for larger values of $\nu$ are similar). In Setting 1, the non-tapered and tapered methods have a very large relative bias at almost every frequency. They estimate far too much power at higher frequencies, due to the fact that imputing with zeros produces fields that are rougher than the underlying process. In contrast, the periodic embedding methods have small bias. In Setting 2, the non-tapered and tapered biases improve, but are still larger than the periodic embedding biases, especially for low frequencies. The relative biases for non-tapered and tapered methods are similar in Setting 3 and are still larger than the periodic embedding relative biases. Though not shown here, the biases for $\tau = 1.1$ and $1.3$ are similar. The parametric filters serve to reduce the bias compared to not filtering. The periodic embedding methods have a small bias near $\omega = (0,0)$; based on the accuracies shown in the numerical studies, this bias is likely due to smoothing bias because of the sharply-peaked spectra near the origin. Imputing nonperiodically does not substantially improve the bias in Settings 2 and 3. It does improve bias in Setting 1, but it is not as effective as periodic embedding.

To evaluate mean relative squared error on an equal footing, all methods were computed with a range of choices for $\delta$; the reported results are for the value of $\delta$ that minimized
\begin{align*}
\Big\{ \frac{1}{m} \sum_{\omega \in \F_z } \mbox{MSE}(\omega) \Big\}^{1/2} ,
\end{align*}
the root integrated mean relative squared error over all Fourier frequencies.
\begin{figure}
\centering
\includegraphics[width=0.9\textwidth]{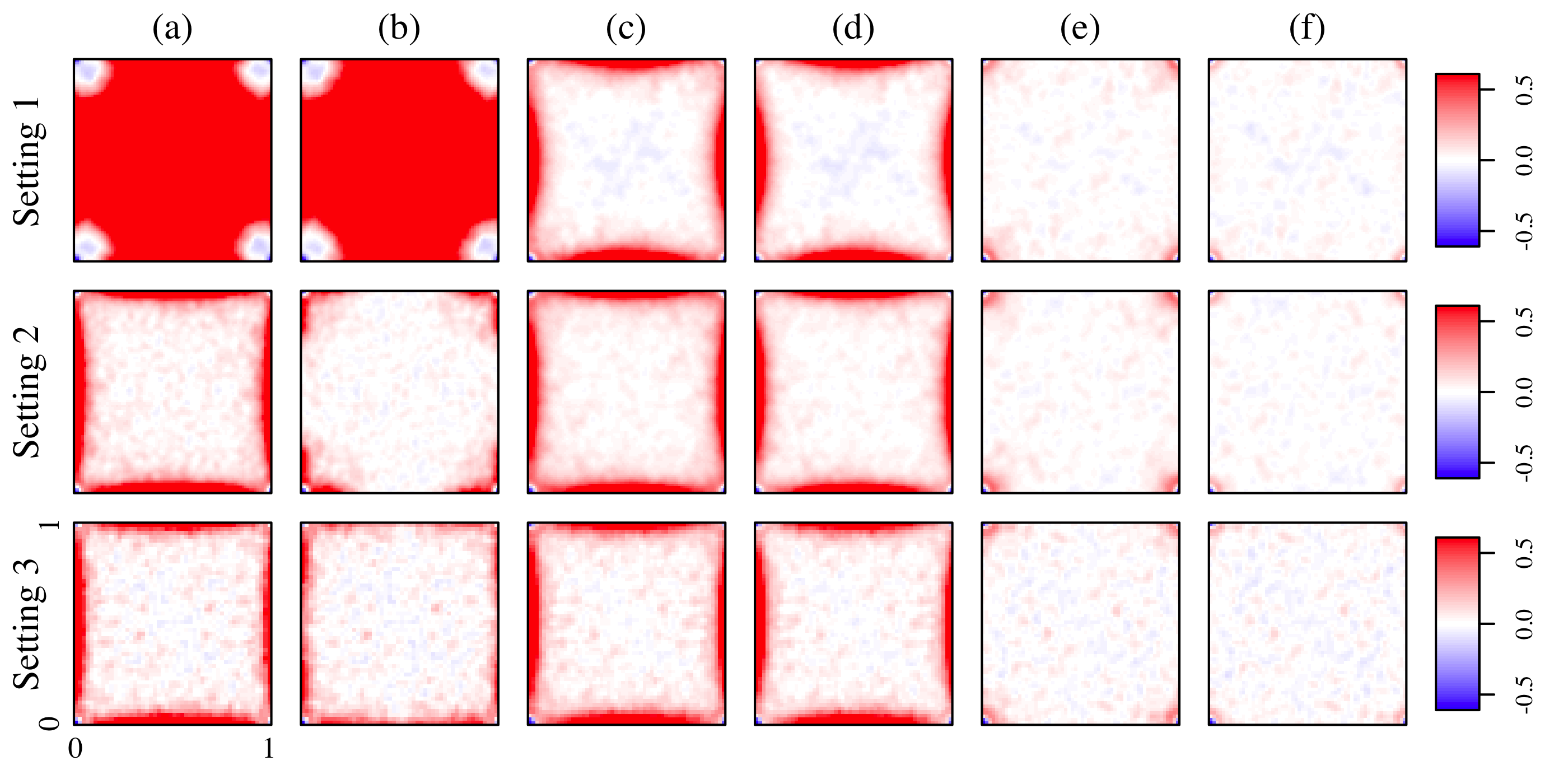}
\caption{Relative bias as a function of frequency for the three missingness settings under $\nu = 1/2$, and six estimation settings: (a) not tapered, (b) tapered, (c) $\tau = 1.2$, no filter, not periodic (d) $\tau = 1.2$, parametric filter, not periodic, (e) $\tau = 1.2$, no filter, periodic, (f) $\tau = 1.2$, parametric filter, periodic}
\label{biasplots}
\end{figure}
Table \ref{tab:RIMSE} contains root integrated mean relative squared error results for the various methods. The periodic embedding methods with $\tau > 1$ are more accurate than both the non-tapered and tapered periodogram estimates in every case. In Setting 1, the non-tapered and tapered estimates are quite poor, likely due to the large biases seen in Figure \ref{biasplots}. For periodic embedding, we see that the values improve when $\tau > 1$ but do not improve beyond $\tau = 1.1$. This is consistent with the numerical studies that showed a small amount of periodic embedding was sufficient. Filtering provides a further improvement, reducing the values by 30-40\%. In Setting 2, the non-tapered and tapered estimates improve substantially, and the periodic embedding methods offer further improvement. Imputing missing values is an improvement, but imputing periodically always gives better results than imputing non-periodically. This can be seen by comparing the $\tau = 1.0$ results to the $\tau > 1$ results and by comparing the periodic to the non-periodic imputation results. In Setting 3, the parametric filter performs similarly to tapering, but periodic embedding with parametric filtering is by far the most accurate method when $\tau > 1$.



\begin{table}
\small
\label{tab:RIMSE}
\centering
\begin{tabular}{cccccccccc}
 & \multicolumn{9}{c}{Missingness Setting} \\
 & 1 & 2 & 3 & 1 & 2 & 3 & 1 & 2 & 3 \\
\cline{2-10}
impute - filter - periodic & \multicolumn{3}{c}{$\nu = 1/2$} & \multicolumn{3}{c}{$\nu = 1$} & \multicolumn{3}{c}{$\nu = 3/2$} \\
\hline
no - no - no & 3.495 & 0.560 & 0.478 & 32.11 & 3.145 & 2.299 & 257.3 & 17.03 & 15.36 \\
no - taper - no & 3.498 & 0.291 & 0.342 & 31.83 & 0.472 & 0.900 & 255.8 & 0.920 & 3.735 \\
$\tau = 1.0$ - no - no & 0.389 & 0.423 & 0.462 & 1.913 & 2.093 & 2.294 & 9.107 & 12.49 & 15.36 \\
$\tau = 1.1$ - no - no & 0.362 & 0.379 & 0.412 & 1.734 & 1.930 & 2.097 & 8.691 & 11.05 & 12.28 \\
$\tau = 1.2$ - no - no & 0.397 & 0.402 & 0.439 & 1.858 & 2.098 & 2.313 & 9.669 & 11.88 & 13.35 \\
$\tau = 1.0$ - yes - no & 0.313 & 0.320 & 0.323 & 1.168 & 1.197 & 1.559 & 6.138 & 8.280 & 8.803 \\
$\tau = 1.1$ - yes - no & 0.284 & 0.296 & 0.296 & 1.001 & 1.041 & 1.260 & 5.775 & 6.469 & 7.511 \\
$\tau = 1.2$ - yes - no & 0.296 & 0.312 & 0.309 & 1.022 & 1.062 & 1.266 & 5.689 & 7.443 & 7.587 \\
$\tau = 1.0$ - no - yes & 0.367 & 0.423 & 0.462 & 1.684 & 2.092 & 2.294 & 8.418 & 12.48 & 15.36 \\
$\tau = 1.1$ - no - yes & 0.208 & 0.219 & 0.259 & 0.238 & 0.268 & 0.326 & 0.288 & 0.297 & 0.382 \\
$\tau = 1.2$ - no - yes & 0.205 & 0.223 & 0.262 & 0.247 & 0.256 & 0.309 & 0.281 & 0.280 & 0.353 \\
$\tau = 1.0$ - yes - yes & 0.253 & 0.319 & 0.323 & 0.908 & 1.195 & 1.559 & 5.348 & 8.266 & 8.803 \\
$\tau = 1.1$ - yes - yes & 0.136 & 0.145 & 0.153 & \textbf{0.096} & \textbf{0.088} & \textbf{0.108} & \textbf{0.141} & 0.141 & 0.166 \\
$\tau = 1.2$ - yes - yes & \textbf{0.133} & \textbf{0.143} & \textbf{0.153} & 0.097 & 0.091 & 0.109 & 0.142 & \textbf{0.137} & \textbf{0.156} \\
\end{tabular}
\caption{Root integrated mean relative squared error results}
\end{table}

\section{Application to Satellite Data}\label{datasection}

To illustrate the practical usefulness of the proposed methods, we analyze a gridded land surface temperature dataset. These data were used recently in \cite{heaton2017methods}, a study comparing various Gaussian process approximations. The data were originally collected by the Moderate Resolution Imaging Spectrometer (MODIS) on board the NASA Terra Satellite. The region is a grid of 500 by 300 locations in the latitudinal range of $34.295$ to $37.068$ and longitudinal range of $-95.912$ to $-91.284$, roughly 450 km by 300 km with grid spacing 1100m in the north/south direction and 900m in the east/west direction. The values in the dataset represent land surface temperature in degrees Celsius. The dataset has $105{,}569$ non-missing values, which are plotted in the top left panel of Figure \ref{predictions}. We can see that there is a distinct trend from the southeast to the northwest corner, so we include a linear trend in the mean function, estimated by generalized least squares.


We have found that $\varepsilon = 0.05$ is a reasonable convergence tolerance criterion, and we choose $B=30$ burn-in iterations. We use a cross-validation procedure to choose the smoothing parameter. A random subset of 30\% of the data is held out, and the iterative methods are run with a range of smoothing parameters, and the parameter that minimizes sum of squared prediction errors was chosen.

In Figure \ref{predictions}, we plot the original data, the conditional expectation, an estimate of the conditional standard deviations, and three conditional simulation plots. The conditional standard deviations are estimated by computing 30 conditional simulations, and computing the root mean squared difference between the conditional expectation and each of the conditional simulations at each pixel.
On average, each conditional simulation took just 2.76 seconds and converged in 25 iterations with the Vecchia preconditioner, and took 15.48 seconds and converged in 159 iterations with the inverse spectrum preconditioner (See Appendix \ref{circulantappendix} for discussion of preconditioners). The iterative spectrum estimation method took 4.86 minutes to converge. While these timings indicate that the analysis is feasible on a large dataset, a zero-infill method is much faster, taking just 0.06 seconds. All timings are done on a 2016 Macbook Pro with 3.3 GHz Intel Core i7 processer (dual core) and with 16GB Memory, running R 3.4.2 linked to Apple's Accelerate BLAS libraries.

\begin{figure}
\centering
\includegraphics[width=\textwidth]{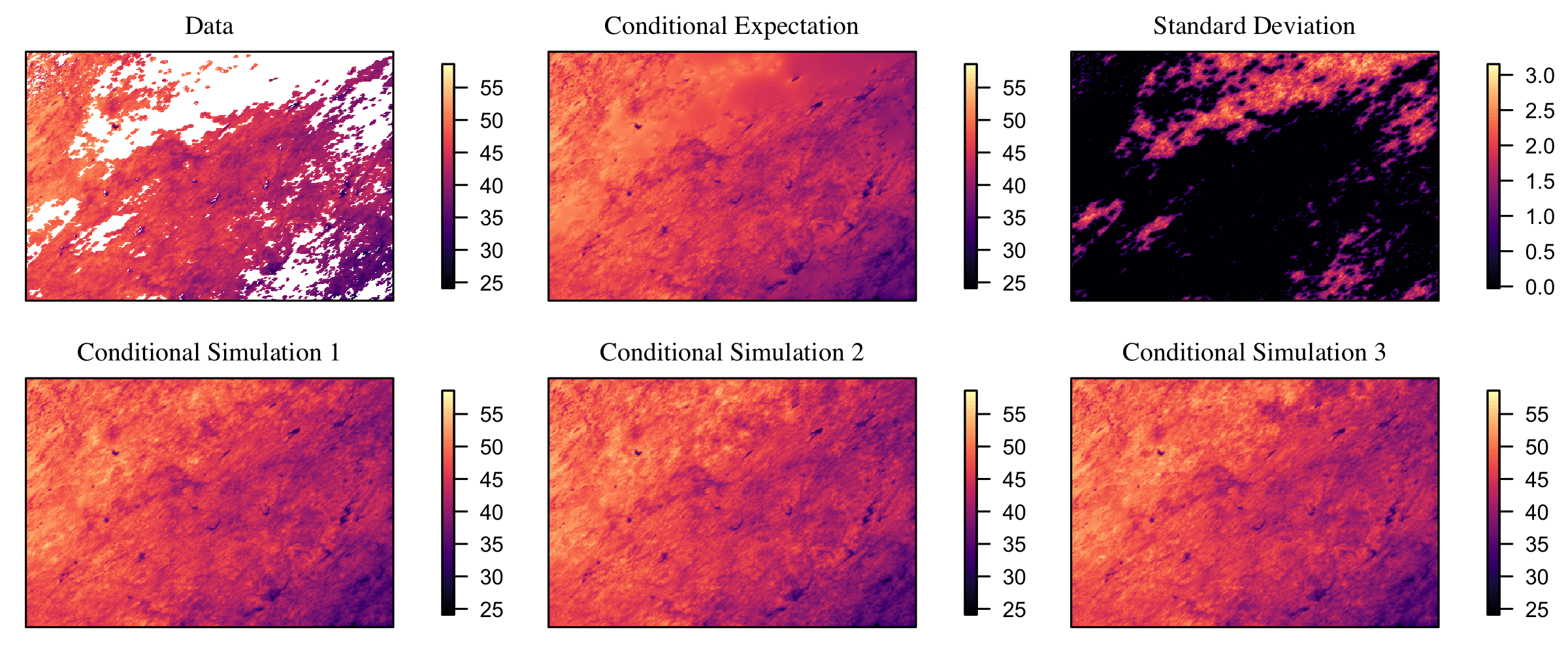}
\caption{Original data, predictions, standard deviations, and three conditional simulations of the missing values}
\label{predictions}
\end{figure}

Visually, the data appear to have a longer correlation lengthscale in the northeast-southwest direction than in the southeast-northwest direction. The estimate of the spectrum returned by the iterative method confirms our visual suspicions, as can be seen in Figure \ref{logspecden} from Appendix \ref{specdenplot}, where the logarithm of the estimated spectrum is plotted. The estimated spectrum shows clear signs of anisotropy in that the spectrum has contours that are not circular. Maximum likelihood estimation of anisotropic models is generally difficult due to optimization over additional parameters. 
In contrast, the nonparametric spectral density estimation methods automatically estimated the anisotropies with no extra computational effort.

The spectral methods described in this paper were included in the \cite{heaton2017methods} comparison project and compared favorably to all of the other methods on all of the prediction and timing metrics, and it was the best performing method for the interval score metric \citep{gneiting2007strictly}, which rewards forecasts that come with small prediction intervals that often contain the predictand. To gain some intuition for this result, we report some results for $(1-\alpha)\%$ prediction intervals based on a Gaussian assumption. In particular, we sort the predictions $(\widehat{Y}_1,\ldots,\widehat{Y}_p)$ to be increasing in the prediction standard deviation, and then report average prediction standard deviations for $(\widehat{Y}_i,\ldots,\widehat{Y}_j)$ for various ranges of indices $i$ and $j$. The results from the periodic spectral methods are compared in Table \ref{pred_comparison} to predictions that use an isotropic Mat\'ern covariance model, with parameters estimated via Vecchia's approximation \citep{vecchia1988estimation}, as implemented in the GpGp R package \citep{guinness2018permutation,gpgp}. Vecchia's approximation applies to parametric models and to both gridded and non-gridded data. We can see that while the two methods do not differ substantially for predictions that the model expects to be uncertain, the periodic spectral methods produce smaller prediction intervals and smaller root mean squared prediction error when the model expects small prediction errors. This is achieved with coverage rates that are larger than those produced by Vecchia's approximation with an isotropic model.

\begin{table}
\centering
\small
\begin{tabular}{cccccccc}
& Index Range & 1 & 501 & 1001 & 2001 & 10001 & 20000 \\
& & 500 & 1000 & 2000 & 10000 & 20000 & 44431 \\
\cline{1-8}
Periodic & Avg Pred SD & \textbf{0.365} & \textbf{0.427} & \textbf{0.482} & \textbf{0.694} & \textbf{1.164} & \textbf{1.88} \\
Spectral & Std.\ Dev.\ & 0.414 & 0.477 & 0.554 & 0.686 & 1.078 & 2.209 \\
 & 80\% Coverage & 81.14 & 82.06 & 83.15 & 84.96 & 85.34 & 73.53 \\
& 90\% Coverage & 86.56 & 89.23 & 89.29 & 91.35 & 92.61 & 84.59 \\
& 95\% Coverage & 91.47 & 92.82 & 93.08 & 94.72 & 95.99 & 91.19 \\
\cline{1-8}
Vecchia & Avg Pred SD & \textbf{0.501} & \textbf{0.548} & \textbf{0.585} & \textbf{0.749} & \textbf{1.198} & \textbf{1.876} \\
& Std.\ Dev.\ & 0.503 & 0.58 & 0.538 & 0.718 & 1.094 & 2.201 \\
& 80\% Coverage & 74.88 & 78.55 & 77.23 & 80.71 & 82.05 & 61.12 \\
& 90\% Coverage & 84.88 & 87.65 & 87.15 & 88.58 & 90.66 & 74.85 \\
& 95\% Coverage & 89.02 & 91.84 & 90.87 & 92.39 & 94.47 & 83.87 \\
\end{tabular}
\caption{\label{pred_comparison} Average prediction standard deviation and coverages for the specified range of predicted values, with the predicted values sorted according to the fitted models' prediction standard deviations. In other words, the first column corresponds to prediction results for the 500 predictions that the model expects to be most certain, and the last column corresponds to the predictions expected to be most uncertain.}
\end{table}

\section{Discussion}\label{discussionsection}


The methods involve choosing the factor by which the lattice should be expanded. We have found that even very small factors that expand the lattice by a few pixels are effective at improving the spectral density estimates. We recommend expanding each dimension by an amount roughly equal to the correlation range in the data. The fact that we expand the lattice in the positive direction--rather than in the negative direction or both positive and negative directions--is not important since we assume a periodic model on the expanded lattice. As with most nonparametric spectral density estimates, the methods involve the choice of a smoothing parameter. We have not attempted to provide any new methods for selecting smoothing parameters, as this issue has been well-studied in the literature. However, we note that the parametric filtering methods serve to flatten the periodogram, which makes the estimates less sensitive to the choice of smoothing parameter. In our application to land surface temperature data, we used a cross-validation procedure to select the smoothing parameter. Though we have chosen $L=1$ imputation per iteration in every example, the methods allow for $L>1$. We suspect that choosing $L>1$ would drive the iterative methods to converge in fewer iterations but incur a higher computational cost per iteration. Examining the details of this tradeoff would be an interesting study. It may be advantageous to use $L>1$ if the conditional simulations can be computed in parallel.

While many large datasets involve spatially gridded observations, we acknowledge that there is also a need for methods for analyzing non-gridded data. The nonparametric methods described in this paper may prove useful for analyzing non-gridded data as well; in fact \cite{nychka2015multiresolution} have a framework for analyzing non-gridded data that includes a lattice process as a model component. We have considered stationary models here. Stationary models can be used as components in nonstationary models \citep{fuentes2002spectral}, and so the methods developed here could potentially be extended to be used for local nonparametric estimation of nonstationary models as well. The paper contains some theoretical results about the iterative procedure, but proving that the iterative algorithm converges remains elusive, partly due to pathological cases in the observed vector $U$, but this is an important area of future work.

\vskip24pt

\begin{center}
\textbf{\Large Acknowledgements}
\end{center}

Guinness's research is supported by NSF grant No.\ 1613219 and NIH grant No.\ R01ES027892.

\appendix

\section{Circulant Embedding and Inverse Spectrum Preconditioner}\label{circulantappendix}

To see how the conditional simulations of $W$ given $U$ can be computed efficiently, define $R$ to be the covariance matrix for $(U,W)$ under covariance function $R(\cdot)$, and partition $R$ as
\begin{align*}
R = \begin{pmatrix} A & B \\ B^T & C \end{pmatrix},
\end{align*}
where $A$ and $C$ are the covariance matrices for the observations $U$ and missing values $W$, respectively. The conditional expectation for $W$ given $U$ is $\widetilde{E}(W|U) = B^T A^{-1} U$. The most demanding computational step for obtaining $\widetilde{E}(W|U)$ is solving the linear system $A x = U$. Preconditioned conjugate gradient methods for solving linear systems \citep{greenbaum1997iterative} are efficient when the forward multiplication $A x$ can be computed efficiently, and when we can find a matrix $M$, called the preconditioner, for which $M A \approx I$ and for which $M x$ can be computed efficiently. Below, we describe how circulant embedding can be used to compute the forward multiplication $Ax$ efficiently. In practice, we have found that a preconditioner based on Vecchia's Gaussian process approximation \citep{vecchia1988estimation} is effective and fast for the problems we have studied. This preconditioner was proposed in \cite{stroud2017bayesian}. At the end of this section, we give details about another preconditioner based on a submatrix of the inverse of $R^{-1}$.


Suppose that $Q$ is an $m \times m$ nested block circulant matrix. Nested block circulant includes the special cases of circulant, arising from a periodic and stationary covariance in one dimension, and block circulant with circulant blocks, arising from a periodic and stationary covariance in two dimensions. The matrix $Q$ can be written as $Q = F D F^\dag$, where $F$ is the discrete Fourier transform matrix, and $D$ is a diagonal matrix with the eigenvalues on the diagonal. Because of the discrete Fourier matrix representation, one can multiply $Q v$ in $O( m \log m )$ time and $O(m)$ memory by taking the discrete Fourier transform of $v$ (i.e.\ \verb!fft(v)! in R), then multiplying the entries of the resultant vector pointwise by the eigenvalues in $D$, then taking an inverse discrete Fourier transform of the result.

The multiplication $A x$ can be computed efficiently by embedding the multiplication inside of
\begin{align*}
R \begin{pmatrix} x \\ 0 \end{pmatrix} = \begin{pmatrix} A & B \\ B^T & C \end{pmatrix} \begin{pmatrix} x \\ 0 \end{pmatrix} = \begin{pmatrix} A x \\ B^T x \end{pmatrix}.
\end{align*}
Then the approriate entries $A x$ can be extracted, and the unnecessary entries $B^T x$ can be discarded. Note that $R$ is not nested block circulant, but there exists a row-column permutation of $R$ that is nested block circulant. Let $P$ denote the permutation matrix such that $Q = P R P^T$ is nested block circulant. Then the multiplication can be performed as
\begin{align*}
R \begin{pmatrix} x \\ 0 \end{pmatrix} = P^T Q P \begin{pmatrix} x \\ 0 \end{pmatrix}.
\end{align*}
Thus, the multiplication can be carried out by an appropriate reordering of $[x \,\,\, 0]$ in $O(m)$ time, then an $O(m \log m)$ time multiplication by nested block circulant $Q$, then an $O(m)$ time reordering of the result.

The preconditioner $M = (A - B C^{-1} B^T)^{-1}$ is a submatrix of $R^{-1}$. Here, we describe how the multiplication $M x$ can be performed efficiently without computing the entries of $M$. The inverse of $R$ is a permutation of a nested block circulant matrix and can be written as
\begin{align*}
R^{-1} = \begin{pmatrix}
(A - B C^{-1} B^T)^{-1} & -(A - B C^{-1} B^T)^{-1} B C^{-1} \\
- C^{-1} B^T ( A - B C^{-1} B^T)^{-1} & (C - B^T A^{-1} B )^{-1}
\end{pmatrix}
= P^T F D^{-1} F^\dag P.
\end{align*}
This means that the multiplication $M x$ can be embedded in the larger multiplication
\begin{align*}
R^{-1} \begin{pmatrix} x \\ 0 \end{pmatrix} = \begin{pmatrix} (A - B C^{-1} B^T )^{-1} x \\ - C^{-1} B^T(A - BC^{-1}B^T)^{-1} x \end{pmatrix} = P^T F D^{-1} F^\dag P \begin{pmatrix} x \\ 0 \end{pmatrix},
\end{align*}
and the multiplication can be carried out in $O(m \log m)$ time and $O(m)$ memory by a sequence of reorderings, discrete Fourier transforms, and pointwise multiplications.

\section{Software}

The methods are implemented in an R package titled ``npspec'' and made freely available through the author's public github page \verb!https://github.com/joeguinness/npspec!. The main function for estimating the spectrum is \verb!iterate_spec!, which takes the following arguments: \verb!y!, a matrix of data values; \verb!observed!, a logical-valued matrix of the same size as \verb!y!, with TRUE entries indicating non-missing values, and FALSE entries indicating missing values; \verb!embed_fac!, which is the embedding factor $\tau$; \verb!burn_iters!, the number of burn-in iterations (described below); \verb!par_spec_fun!, which currently takes values \verb!spec_AR1!, indicating that the function should filter using the spectrum in \eqref{spec_AR1}, or \verb!FALSE!, indicating that no parametric filtering should be used; \verb!kern_parm!, the kernel smoothing parameter $\delta$; \verb!precond_method!, either ``\verb!fft!'' for the inverse spectrum preconditioner or ``\verb!Vecchia!'' for the Vecchia preconditioner; and \verb!m!, the number of neighbors when the Vecchia preconditioner is used. In addition to returning an estimated spectrum \verb!spec!, \verb!iterate_spec! returns a conditional expectation map filling in the missing values, a conditional simulation map, and a Vecchia approximation to the likelihood.

Additionally, there is a function \verb!condexp! for computing predictions at the locations of the missing values. It has arguments \verb!y!, \verb!observed!, and \verb!spec!, the spectrum to be used for prediction. There is also a function \verb!condsim! for computing \verb!L! conditional simulations given data \verb!y!, observation locations \verb!observed!, and a spectrum \verb!spec!.

\section{Plot of Estimated Temperature Spectral Density}\label{specdenplot}

In Figure \ref{logspecden}, we plot the estimated spectral density from the surface temperature dataset.

\begin{figure}
\centering
\includegraphics[width=0.5\textwidth]{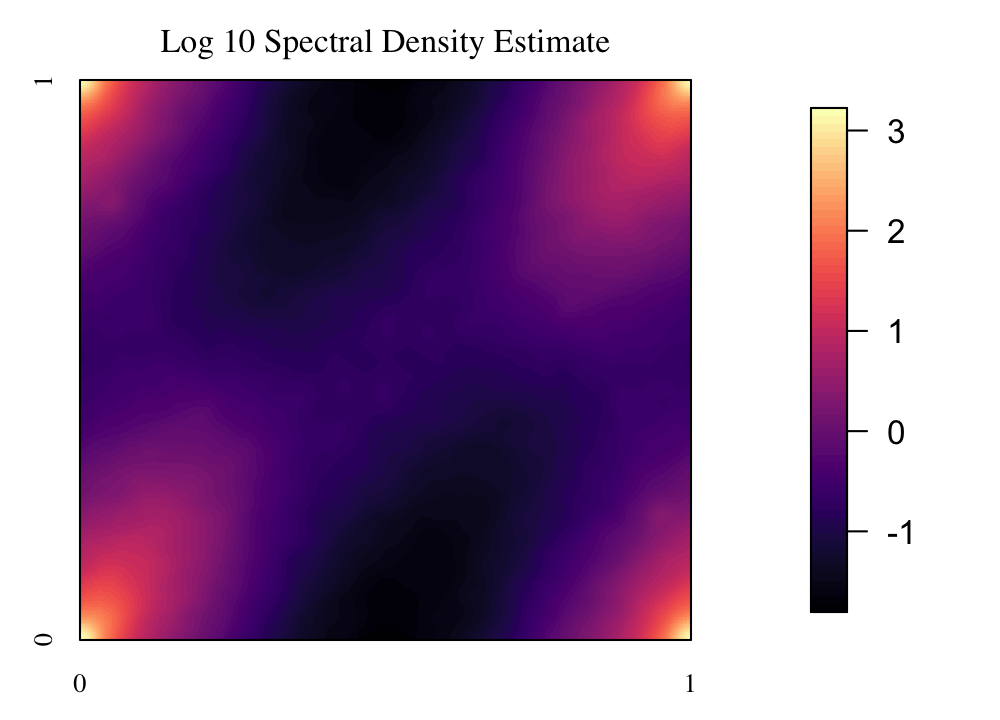}
\caption{Log base 10 of spectral density estimate.}
\label{logspecden}
\end{figure}

\section{Proofs of Theoretical Results}\label{proofs}

\begin{lemma}\label{compacttheorem}
If $f_1 = f$, $K(h) = 0$ for all $\| h \| > h_0$, and $|z_j - y_j| > h_0$ for all $j \in \{1,\ldots,d\}$, then for every $\nu,\omega \in \F_z$,
\begin{align*}
    \widetilde{f}_1(\nu,\omega) - f(\nu,\omega) = 0.
\end{align*}
\end{lemma}
\begin{proof}
We have $f(\nu,\omega) = g(\nu)^\dag R g(\omega)$, and so $\widetilde{f}_1(\nu,\omega) - f(\nu,\omega) = g(\nu)^\dag (S - R) g(\omega)$. The matrix $S-R$ can be written as
\begin{align*}
S - R =
\begin{pmatrix}
K - A & (K - A)A^{-1}B \\
B^T A^{-1}(K-A) & B^T A^{-1}(K-A) A^{-1}B
\end{pmatrix}.
\end{align*}
It sufficies to show that $K = A$ in order to establish the result. According to the multidimensional Poisson summation formula, we can relate $K(\cdot)$ and $R(\cdot)$ as \citep[Lemma 1, Proved in Appendix C]{guinness2016circulant}
\begin{align}\label{poissonformula2}
R(x_1 - x_2) = \sum_{k \in \Z^d} K\{ x_1 - x_2 + (z \circ k) \} = K(x_1 - x_2) + \sum_{k \in \Z^d \setminus 0} K\{x_1 - x_2 + (z \circ k)\},
\end{align}
where $z \circ k := (z_1 k_1, \ldots, z_d k_d)$. For any $x_1,x_2 \in \J_y$ (the observation lattice), we have $|x_{1j} - x_{2j}| < y_j$ for every $j = 1,\ldots,d$. Thus if $k_j \in \Z \setminus 0$, $|x_{1j} - x_{2j} + z_j k_j| > |z_j - y_j | > h_0$. Thus at least one element of $x_1 - x_2 + (z \circ k)$ has absolute value greater than $h_0$ when $k \in \Z^d \setminus 0$, and so $\| x_1 - x_2 + (z\circ k) \| > h_0$ for all $k \in \Z^d \setminus 0$, implying that all terms in the sum in \eqref{poissonformula2} must be zero. This gives us $R(x_1 - x_2) = K(x_1 - x_2)$ for any $x_1,x_2 \in \J_y$, and so $K = A$.
\end{proof}

\vspace{12pt}

\begin{lemma}\label{theoremB}
If $f$ has $p$ continuous partial derivatives, $f_1 = f$, $y_j = O( n^{1/d} )$, and $z_j = O(\tau y_j)$ for $\tau > 1$, then for every $\nu,\omega \in \F_z$,
\begin{align*}
\widetilde{f}_1(\nu,\omega) - f(\nu,\omega)  = O(n^{-p/d+1} ).
\end{align*}
\end{lemma}
\begin{proof}

As in the proof of Lemma \ref{compacttheorem}, $\widetilde{f}_1(\nu,\omega) - f(\nu,\omega)= g(\nu)^\dag (S-R) g(\omega)$. Partitioning the vector $g(\nu)$ as $(g_1(\nu),g_2(\nu))$ according to the same partition as $(U,W)$, we have
\begin{align*}
\widetilde{f}_1(\nu,\omega) - f(\nu,\omega) &=\begin{pmatrix}
g_1(\nu)^\dag & g_2(\nu)^\dag
\end{pmatrix}
\begin{pmatrix}
K-A & (K-A)A^{-1}B \\
B^T A^{-1}(K-A) & B^T A^{-1}(K-A)A^{-1}B
\end{pmatrix}
\begin{pmatrix}
g_1(\omega) \\ g_2(\omega)^\dag
\end{pmatrix} \\
&= \{g_1(\nu)^\dag + g_2(\nu)^\dag B^T A^{-1} \}(K-A)\{g_1(\omega) + A^{-1}B g_2(\omega) \},
\end{align*}
and so the difference can be bounded as
\begin{align*}
|\widetilde{f}_1(\nu,\omega) - f(\nu,\omega)| &= |\{ g_1(\nu)^\dag + g_2(\nu)^\dag B^T A^{-1} \}(K-A)\{g_1(\omega) + A^{-1}B g_2(\omega) \}| \\
 & \leq \| g_1(\nu) + A^{-1}B g_2(\nu) \|_2 \, \| K-A \|_2  \, \| g_1(\omega) + A^{-1}B g_2(\omega) \|_2.
\end{align*}
We will consider each term in turn. Let $\rho(M)$ denote the spectral radius of symmetric matrix $M$. Then
\begin{align*}
\| g_1(\nu) + A^{-1}B g_2(\nu) \|_2 &\leq \|g_1(\nu)\|_2 + \| A^{-1} B g_2(\nu) \|_2 \\
 & \leq \sqrt{\frac{n}{m}} + \| A^{-1/2} \|_2 \, \| A^{-1/2} B \|_2 \, \| g_2(\nu) \|_2 \\
 & = \sqrt{\frac{n}{m}} + \rho( A^{-1} )^{1/2} \rho( B^T A^{-1} B )^{1/2} \| g_2(\nu) \|_2 \\
 & \leq \sqrt{\frac{n}{m}} + f_{\min}^{-1/2} \rho( C )^{1/2} \sqrt{\frac{m-n}{m}} \\
 & \leq \sqrt{\frac{n}{m}} + \sqrt{\frac{f_{\max}}{f_{\min}}} \sqrt{\frac{m-n}{m}},
\end{align*}
where in the second to last inequality, we used $\rho(A^{-1}) < \rho(R^{-1}) = f_{\min}^{-1}$, and $C - B^T A^{-1} B$ positive definite, so the largest eigenvalue of $B^T A^{-1} B$ is smaller than the largest eigenvalue of $C$, which is smaller than the largest eigenvalue of $R$, $f_{\max}$.

The previous inequality did not depend on $\nu$, so it holds for $\| g_1(\omega) + A^{-1} B g_2(\omega) \|_2$ as well. To bound $\| K - A \|_2$, we use the fact that for symmetric matrices $\| K - A \|_2 = \rho( K - A ) < \| K - A \|_1$, where
\begin{align*}
\| K - A \|_1 &= \max_{i} \sum_{j=1}^n | K_{ij} - A_{ij} | \\
    &= \max_{x_1 \in \J_y } \sum_{x_2 \in \J_y } \Big| K( x_1 - x_2 ) - R(x_1 - x_2) |\\
    &= \max_{x_1 \in \J_y } \sum_{x_2 \in \J_y } \Big| K(x_1 - x_2 ) - \sum_{k \in \Z^d} K\{ x_1 - x_2 + (k \circ z) \} \Big| \\
    &= \max_{x_1 \in \J_y } \sum_{x_2 \in \J_y } \Big| \sum_{k \neq \bm{0} } K\{ x_1 - x_2 + (k \circ z) \} \Big| \\
    & \leq \max_{x_1 \in \J_y } \sum_{x_2 \in \J_y }  \sum_{k \neq \bm{0} } \Big| K\{ x_1 - x_2 + (k \circ z) \} \Big|
\end{align*}
The third equality uses the multidimensional Poisson summation formula referenced in \eqref{poissonformula2}. By assumption, for $k = (k_1,\ldots,k_d) \neq \bm{0}$ and any $x_1,x_2 \in \J_y$,
\begin{align*}
\max_{1\leq j \leq d} |x_{1j} - x_{2j} + z_j k_j| \geq \min_{1\leq j \leq d} | z_j - y_j |.
\end{align*}
This is because $k_j \neq 0$ for at least one $j$. Define $\ell_{min} = \min_{1\leq j \leq d} |z_j - y_j|$, which is the embedding distance in the dimension with the smallest amount of embedding. By assumption, we have
\begin{align*}
\ell_{\min} > \min_{1\leq j \leq d} \lfloor (\tau - 1)y_j \rfloor > \min_{1\leq j \leq d} \lfloor (\tau - 1)a_j n^{1/d} \rfloor.
\end{align*}
This means the sum does not contain any terms $K(h)$ for which $\max_j |h_j| < \ell_{\min}$. Define the set $G_\ell = \{h : \max_{j} |h_j| = \ell \}$, which is a hollowed out cube on $\Z^d$ and has size $|G_\ell| = (2\ell+1)^d - (2\ell-1)^d = O(\ell^{d-1})$. Using this notation, the sum can be bounded as
\begin{align*}
\| K - A \|_1 \leq \sum_{\ell = \ell_{\min}}^\infty \sum_{ h \in G_\ell  } |K(h)|.
\end{align*}

Lemma 9.5 in \cite{korner1989fourier} states that if $f$ has $p$ continuous partial derivatives on $\T^d$, with maximum $p$th partial derivative $Q_p(f)$, then
\begin{align*}
|K(h)| \leq Q_p(f)  |h_j|^{-p}
\end{align*}
for every $1 \leq j \leq d$, and so we can use the bound $|K(h)| \leq Q_p(f) (\max_{1\leq j \leq d} |h_j|)^{-p}$. This gives us an explicit bound
\begin{align*}
\| K - A \|_1 & \leq \sum_{\ell = \ell_{\min}}^\infty \sum_{h \in G_\ell} \frac{Q_p(f)}{\ell^p} =
\sum_{\ell = \ell_{\min}}^\infty \Big\{ (2\ell+1)^d - (2\ell-1)^d \Big\} \frac{Q_p(f)}{\ell^p} \\
&= \sum_{\ell = \ell_{\min}}^\infty Q_p(f) \frac{ a_{d-1}(\ell) }{\ell^p},
\end{align*}
where $a_{d-1}(\ell)$ is a polynomial of degree $d-1$ in $\ell$. Then we have
\begin{align*}
\| K - A \|_1  = O( \ell_{\min}^{-p+d} ) = O( n^{-p/d+1} )
\end{align*}
since the largest exponent in $a_{d-1}(\ell)/\ell^p$ is $-p+d-1$. Combining this with $\| g_1(\nu) + A^{-1}Bg_2(\nu) \| = O(1)$ gives the desired result.
\end{proof}

\vskip12pt

\setcounter{theorem}{0}
\begin{theorem} Let $f$ have $p$ continuous partial derivatives, and assume the same conditions on the observation and embedding lattice as in Theorem \ref{theoremB}. Define $\Delta_1 = \max_{\omega \in [0,1]^d}| f(\omega) - f_1(\omega)|$. Then for every $\nu,\omega \in \F_z$,
\begin{align*}
\widetilde{f}_1(\nu,\omega) - f(\nu,\omega) = O(n^{-p/d+1}) + O\Big(\Delta_1 \sqrt{\frac{m-n}{m}}   \Big),
\end{align*}
meaning that the difference contains two terms with the respective rates.
\end{theorem}

\begin{proof}
Define the matrix $S_1$ as
\begin{align*}
S_1 = \begin{pmatrix}
A & A A_1^{-1} B_1 \\
B_1^T A_{1}^{-1}A & C_1 - B_1^T A_1^{-1} B_1 + B_1^T A_1^{-1} A A_1^{-1} B_1
\end{pmatrix},
\end{align*}
which differs from $S$ in that $K$ in $S$ is replaced by $A$ in $S_1$. The difference $\widetilde{f}_1(\nu,\omega) - f(\nu,\omega)$ can be written as
\begin{align}\label{telescope2}
g(\nu)^\dag(S-R)g(\omega) = g(\nu)^\dag (S - S_1) g(\omega) + g(\nu)^\dag( S_1 - R_1 + R_1 - R)g(\omega).
\end{align}
The first term in \eqref{telescope2} is
\begin{align*}
\delta_1 &:=\begin{pmatrix}
g_1(\nu)^\dag & g_2(\nu)^\dag
\end{pmatrix}
\begin{pmatrix}
K-A & (K-A)A_1^{-1}B_1 \\
B_1^T A_1^{-1}(K-A) & B_1^T A_1^{-1}(K-A)A_1^{-1}B_1
\end{pmatrix}
\begin{pmatrix}
g_1(\omega) \\ g_2(\omega)^\dag
\end{pmatrix} \\
&= \{g_1(\nu)^\dag + g_2(\nu)^\dag B_1^T A_1^{-1} \}(K-A)\{g_1(\omega) + A_1^{-1}B_1 g_2(\omega) \}.
\end{align*}
This expression has a similar form to that which appears in the proof of Theorem \ref{theoremB}. As before, we need bounds for $\| g_1(\omega) + A_1^{-1}B_1 g_2(\omega) \|_2$ and $\| K-A\|_2$ in order to bound $\delta_1$. The proof for the bound on $\| K - A \|_2$ is identical to that in Theorem \ref{theoremB}, and the proof for the bound on $\|g_1(\omega) + A_1^{-1}B_1 g_2(\omega) \|_2$ is similar, although $f$ is replaced by $f_1$, which does not change the overall result that the first term in \eqref{telescope2} is $O(n^{-p/d+1})$.

To shorten the equations to follow, write $M_1 = A_1^{-1} B_1$. The second term in \eqref{telescope2} is
\begin{align*}
\delta_2 &:=\begin{pmatrix}
g_1(\nu)^\dag & g_2(\nu)^\dag
\end{pmatrix}
\begin{pmatrix}
A-A_1 & (A-A_1)M_1 \\
M_1^T(A-A_1) & M_1^T (A-A_1)M_1
\end{pmatrix}
\begin{pmatrix}
g_1(\omega) \\ g_2(\omega)^\dag
\end{pmatrix} + f_1(\nu,\omega) - f(\nu,\omega) \\
&= \{ g_1(\nu)^\dag + g_2(\nu)^\dag M_1^T \}(A-A_1)\{ g_1(\omega) + M_1 g_2(\omega)\} + f_1(\nu,\omega) - f(\nu,\omega).
\end{align*}
Define the discrete Fourier transform matrix $F$ to have $(j,k)$ entry $m^{-1/2} \exp(i \omega_k \cdot x_j)$, where $\omega_k$ is a Fourier frequency in $\F_z$, and $x_j$ is a location in $\J_z$. Partition the discrete Fourier transform matrix $F$ into rows for the observations and missing values as $F^\dag = [G^\dag \, H^\dag]$. We have $R = F D F^\dag$, where $D$ is diagonal with entries $f(\nu,\omega)$. This gives $A = G D G^\dag$, and likewise $A_1 = G D_1 G^\dag$, where $D_1$ is diagonal with entries $f_1(\nu,\omega)$. Then $\delta_2$ can then be written as
\begin{align*}
\delta_2 = \{ g_1(\nu)^\dag G + g_2(\nu)^\dag M_1^T G \}(D-D_1)\{ G^\dag g_1(\omega) + G^\dag M g_2(\omega)\} + f_1(\nu,\omega) - f(\nu,\omega).
\end{align*}
Note that $I = G^\dag G + H^\dag H$. Since $g_1(\nu)^\dag$ is a row of $G^\dag$, and $g_2(\nu)^\dag$ is the same row of $H^\dag$ we have
\begin{align*}
g_1(\nu)^\dag G + g_2(\nu)^\dag H = e(\nu),
\end{align*}
where $e(\nu) = 1$ for the entry corresponding to $\nu$ and $0$ otherwise. This gives
\begin{align*}
\delta_2 = \{ e(\nu)^T + g_2(\nu)^\dag(M_1^T G - H) \}(D-D_1)\{ e(\omega) + (G^\dag M_1 - H^\dag) g_2(\omega)\} + f_1(\nu,\omega) - f(\nu,\omega).
\end{align*}
We can see now that since $e(\nu)^T(D-D_1)e(\omega) = f(\nu,\omega)-f_1(\nu,\omega)$, there is a cancellation, giving
\begin{align*}
\delta_2 &= g_2(\nu)^\dag(M_1^T G - H)(D-D_1)e(\omega) + e(\nu)^T(D-D_1)(G_1^\dag M_1 - H^\dag)g_2(\omega) \\
&+ g_2(\nu)^\dag(M_1^T G - H)(D-D_1)(G_1^\dag M_1 - H^\dag)g_2(\omega).
\end{align*}
This cancellation is the key step. Using matrix norm inequalities, we have
\begin{align*}
|\delta_2| &\leq \| g_2(\nu) \|_2 \| M_1^T G - H \|_2 \| D - D_1 \|_2 \| e(\omega) \|_2 \\
           & + \| g_2(\omega) \|_2 \| M_1^T G - H \|_2 \| D - D_1 \|_2 \| e(\nu) \|_2 \\
           & + \| g_2(\omega) \|_2 \| g_2(\nu) \| \| M_1^T G - H \|_2^2 \| D - D_1 \|_2 \| e(\nu) \|_2
\end{align*}
Since $g_2(\omega)$ is of length $n_2$ and has entries $n^{-1/2} \exp(2\pi i\omega \cdot x)$, $\| g_2(\omega) \|_2 = \sqrt{n_2/n}$. Clearly, $\| e(\omega) \| = 1$ because of its definition, and $\| D - D_1 \|_2 = \Delta_1$ because both $D$ and $D_1$ are diagonal with diagonal entries holding $f(\nu,\nu)$ and $f_1(\nu,\nu)$, respectively. This leaves
\begin{align*}
\| M_1^T G - H \|_2^2 = \rho\Big( (M_1^T G - H)( G^\dag M_1 - H^\dag ) \Big) = \rho\Big( M_1^T M_1 + I \Big),
\end{align*}
because $G G^\dag = I$, $H H^\dag = I$, $H G^\dag = 0$, and $G H^\dag = 0$. Thus, the squared 2-norm is 1 plus the largest eigenvalue of $M_1^T M$, which is
\begin{align*}
\| M_1^T M_1 \|_2 = \| M_1 \|_2^2 = \| A_1^{-1} B_1 \|_2^2 \leq \frac{f_{1,\max}}{f_{1,\min}},
\end{align*}
with the last inequality following from the proof of Theorem \ref{theoremB}. Bringing this all together gives
\begin{align*}
|\delta_2| &\leq 2 \sqrt{\frac{m-n}{m}}\sqrt{1 + \frac{f_{1,\max}}{f_{1,\min}}}\Delta_1 + \frac{m-n}{m}\Big( 1 + \frac{f_{1,\max}}{f_{1,\min}} \Big) \Delta_1 \\
 &= O\Big( \sqrt{\frac{m-n}{m}}  \Delta_1 \Big),
\end{align*}
establishing the second term of the Theorem.

\end{proof}

\bibliography{refs}{}

\begin{thebibliography}{}

\bibitem[Bandyopadhyay and Lahiri, 2009]{bandyopadhyay2009asymptotic}
Bandyopadhyay, S. and Lahiri, S. (2009).
\newblock Asymptotic properties of discrete fourier transforms for spatial
  data.
\newblock {\em Sankhy{\=a}: The Indian Journal of Statistics, Series A
  (2008-)}, pages 221--259.

\bibitem[Bandyopadhyay et~al., 2015]{bandyopadhyay2015frequency}
Bandyopadhyay, S., Lahiri, S.~N., Nordman, D.~J., et~al. (2015).
\newblock A frequency domain empirical likelihood method for irregularly spaced
  spatial data.
\newblock {\em The Annals of Statistics}, 43(2):519--545.

\bibitem[Chow and Grenander, 1985]{chow1985sieve}
Chow, Y.-S. and Grenander, U. (1985).
\newblock A sieve method for the spectral density.
\newblock {\em The Annals of Statistics}, pages 998--1010.

\bibitem[Dahlhaus and K{\"u}nsch, 1987]{dahlhaus1987edge}
Dahlhaus, R. and K{\"u}nsch, H. (1987).
\newblock Edge effects and efficient parameter estimation for stationary random
  fields.
\newblock {\em Biometrika}, 74(4):877--882.

\bibitem[Deb et~al., 2017]{deb2017asymptotic}
Deb, S., Pourahmadi, M., Wu, W.~B., et~al. (2017).
\newblock An asymptotic theory for spectral analysis of random fields.
\newblock {\em Electronic Journal of Statistics}, 11(2):4297--4322.

\bibitem[Fuentes, 2002]{fuentes2002spectral}
Fuentes, M. (2002).
\newblock Spectral methods for nonstationary spatial processes.
\newblock {\em Biometrika}, 89(1):197--210.

\bibitem[Fuentes, 2007]{fuentes2007approximate}
Fuentes, M. (2007).
\newblock Approximate likelihood for large irregularly spaced spatial data.
\newblock {\em Journal of the American Statistical Association},
  102(477):321--331.

\bibitem[Gneiting and Raftery, 2007]{gneiting2007strictly}
Gneiting, T. and Raftery, A.~E. (2007).
\newblock Strictly proper scoring rules, prediction, and estimation.
\newblock {\em Journal of the American Statistical Association},
  102(477):359--378.

\bibitem[Greenbaum, 1997]{greenbaum1997iterative}
Greenbaum, A. (1997).
\newblock {\em Iterative methods for solving linear systems}.
\newblock SIAM.

\bibitem[Guinness, 2018]{guinness2018permutation}
Guinness, J. (2018).
\newblock Permutation and grouping methods for sharpening {G}aussian process
  approximations.
\newblock {\em Technometrics}, (in press).

\bibitem[Guinness and Fuentes, 2017]{guinness2016circulant}
Guinness, J. and Fuentes, M. (2017).
\newblock Circulant embedding of approximate covariances for inference from
  {G}aussian data on large lattices.
\newblock {\em Journal of Computational and Graphical Statistics}, 26(1).

\bibitem[Guinness and Katzfuss, 2018]{gpgp}
Guinness, J. and Katzfuss, M. (2018).
\newblock {\em GpGp: Fast Gaussian Process Computation Using Vecchia's
  Approximation}.
\newblock R package version 0.1.0, Available at
  https://CRAN.R-project.org/package=GpGp.

\bibitem[Guyon, 1982]{guyon1982parameter}
Guyon, X. (1982).
\newblock Parameter estimation for a stationary process on a d-dimensional
  lattice.
\newblock {\em Biometrika}, 69(1):95--105.

\bibitem[Heaton et~al., 2017]{heaton2017methods}
Heaton, M.~J., Datta, A., Finley, A., Furrer, R., Guhaniyogi, R., Gerber, F.,
  Gramacy, R.~B., Hammerling, D., Katzfuss, M., and Lindgren, F. (2017).
\newblock Methods for analyzing large spatial data: A review and comparison.
\newblock {\em arXiv preprint arXiv:1710.05013}.

\bibitem[Heyde and Gay, 1993]{heyde1993smoothed}
Heyde, C. and Gay, R. (1993).
\newblock Smoothed periodogram asymptotics and estimation for processes and
  fields with possible long-range dependence.
\newblock {\em Stochastic Processes and their Applications}, 45(1):169--182.

\bibitem[K{\"o}rner, 1989]{korner1989fourier}
K{\"o}rner, T.~W. (1989).
\newblock {\em Fourier {A}nalysis}.
\newblock Cambridge {U}niversity {P}ress.

\bibitem[Lee, 1997]{lee1997simple}
Lee, T.~C. (1997).
\newblock A simple span selector for periodogram smoothing.
\newblock {\em Biometrika}, 84(4):965--969.

\bibitem[Lee, 2001]{lee2001stabilized}
Lee, T.~C. (2001).
\newblock A stabilized bandwidth selection method for kernel smoothing of the
  periodogram.
\newblock {\em Signal Processing}, 81(2):419--430.

\bibitem[Lee and Zhu, 2009]{lee2009nonparametric}
Lee, T.~C. and Zhu, Z. (2009).
\newblock Nonparametric spectral density estimation with missing observations.
\newblock In {\em Acoustics, Speech and Signal Processing, 2009. ICASSP 2009.
  IEEE International Conference on}, pages 3041--3044. IEEE.

\bibitem[Lim and Stein, 2008]{lim2008properties}
Lim, C.~Y. and Stein, M. (2008).
\newblock Properties of spatial cross-periodograms using fixed-domain
  asymptotics.
\newblock {\em Journal of Multivariate Analysis}, 99(9):1962--1984.

\bibitem[Lindgren et~al., 2011]{lindgren2011explicit}
Lindgren, F., Rue, H., and Lindstr{\"o}m, J. (2011).
\newblock An explicit link between {G}aussian fields and {G}aussian {M}arkov
  random fields: the stochastic partial differential equation approach.
\newblock {\em Journal of the Royal Statistical Society: Series B (Statistical
  Methodology)}, 73(4):423--498.

\bibitem[Little and Rubin, 2014]{little2014statistical}
Little, R.~J. and Rubin, D.~B. (2014).
\newblock {\em Statistical analysis with missing data}, volume 333.
\newblock John Wiley \& Sons.

\bibitem[Matsuda and Yajima, 2009]{matsuda2009fourier}
Matsuda, Y. and Yajima, Y. (2009).
\newblock Fourier analysis of irregularly spaced data on rd.
\newblock {\em Journal of the Royal Statistical Society: Series B (Statistical
  Methodology)}, 71(1):191--217.

\bibitem[Nychka et~al., 2015]{nychka2015multiresolution}
Nychka, D., Bandyopadhyay, S., Hammerling, D., Lindgren, F., and Sain, S.
  (2015).
\newblock A multiresolution {G}aussian process model for the analysis of large
  spatial datasets.
\newblock {\em Journal of Computational and Graphical Statistics},
  24(2):579--599.

\bibitem[Ombao et~al., 2001]{ombao2001simple}
Ombao, H.~C., Raz, J.~A., Strawderman, R.~L., and Von~Sachs, R. (2001).
\newblock A simple generalised crossvalidation method of span selection for
  periodogram smoothing.
\newblock {\em Biometrika}, 88(4):1186--1192.

\bibitem[Pawitan and O'{S}ullivan, 1994]{pawitan1994nonparametric}
Pawitan, Y. and O'{S}ullivan, F. (1994).
\newblock Nonparametric spectral density estimation using penalized {W}hittle
  likelihood.
\newblock {\em Journal of the American Statistical Association},
  89(426):600--610.

\bibitem[Politis and Romano, 1995]{politis1995bias}
Politis, D.~N. and Romano, J.~P. (1995).
\newblock Bias-corrected nonparametric spectral estimation.
\newblock {\em Journal of {T}ime {S}eries {A}nalysis}, 16(1):67--103.

\bibitem[Stein, 1995]{stein1995fixed}
Stein, M.~L. (1995).
\newblock Fixed-domain asymptotics for spatial periodograms.
\newblock {\em Journal of the American Statistical Association},
  90(432):1277--1288.

\bibitem[Stroud et~al., 2017]{stroud2017bayesian}
Stroud, J.~R., Stein, M.~L., and Lysen, S. (2017).
\newblock Bayesian and maximum likelihood estimation for {G}aussian processes
  on an incomplete lattice.
\newblock {\em Journal of Computational and Graphical Statistics},
  26(1):108--120.

\bibitem[Subba~Rao, 2018]{subbarao2018}
Subba~Rao, S. (2018).
\newblock Statistical inference for spatial statistics defined in the fourier
  domain.
\newblock {\em Ann. Statist.}, 46(2):469--499.

\bibitem[Sykulski et~al., 2016]{sykulski2016biased}
Sykulski, A.~M., Olhede, S.~C., and Lilly, J.~M. (2016).
\newblock The de-biased {W}hittle likelihood for second-order stationary
  stochastic processes.
\newblock {\em arXiv preprint arXiv:1605.06718}.

\bibitem[Tanner and Wong, 1987]{tanner1987calculation}
Tanner, M.~A. and Wong, W.~H. (1987).
\newblock The calculation of posterior distributions by data augmentation.
\newblock {\em Journal of the American Statistical Association},
  82(398):528--540.

\bibitem[Vecchia, 1988]{vecchia1988estimation}
Vecchia, A.~V. (1988).
\newblock Estimation and model identification for continuous spatial processes.
\newblock {\em Journal of the Royal Statistical Society. Series B
  (Methodological)}, 50(2):297--312.

\bibitem[Wahba, 1980]{wahba1980automatic}
Wahba, G. (1980).
\newblock Automatic smoothing of the log periodogram.
\newblock {\em Journal of the American Statistical Association},
  75(369):122--132.

\bibitem[Whittle, 1954]{whittle1954stationary}
Whittle, P. (1954).
\newblock On stationary processes in the plane.
\newblock {\em Biometrika}, 41(3/4):434--449.

\bibitem[Zheng et~al., 2009]{zheng2009nonparametric}
Zheng, Y., Zhu, J., and Roy, A. (2009).
\newblock Nonparametric {B}ayesian inference for the spectral density function
  of a random field.
\newblock {\em Biometrika}, 97(1):238--245.

\end{thebibliography}
\bibliographystyle{apalike}

\end{document}